\documentclass[11pt,reqno]{amsart}

\usepackage{amsmath, amssymb, amsthm, amsfonts, cite, enumitem}

\usepackage[margin=25mm]{geometry}
\usepackage{hyperref}
\usepackage{graphics,pstricks}

\newtheorem{theorem}{Theorem}[section]
\newtheorem{proposition}[theorem]{Proposition}

\newtheorem{lemma}[theorem]{Lemma}
\newtheorem{remark}[theorem]{Remark}

\numberwithin{equation}{section}

\newcommand{\Z}{\mathbb{Z}}
\newcommand{\R}{\mathbb{R}}

\newcommand{\n}{\mathbb{N}}

\newcommand{\p}{\mathbb{P}}
\newcommand{\esp}{\mathbb{E}}

\newcommand{\bra}{\langle}
\newcommand{\ket}{\rangle}

\newcommand{\norm}[1]{\left\lVert #1 \right\rVert}

\newcommand{\e}{\mathrm{e}}
\newcommand{\bT}{\mathbf{T}}

\def \simless {\mathbin{\lower 3pt\hbox{$\rlap{\raise 5pt
              \hbox{$\char'074$}}\mathchar"7218$}}}

\date{}

\author{Olivier Bourget, Gregorio R. Moreno Flores, Amal Taarabt}

\title[Anderson Model in a Decaying Potential]{One-dimensional Discrete Anderson Model in a Decaying  Random Potential: from a.c. Spectrum to Dynamical Localization}

\begin{document}

\maketitle
\begin{abstract}
	We consider a one-dimensional Anderson model where the potential decays in average like $n^{-\alpha}$, $\alpha>0$. This simple model is known to display a rich phase diagram with different kinds of spectrum arising as the decay rate $\alpha$ varies.
	
	We review an article of Kiselev, Last and Simon where the authors show a.c. spectrum in the super-critical case $\alpha>\frac12$, a transition from singular continuous to pure point spectrum in the critical case $\alpha=\frac12$, and dense pure point spectrum in the sub-critical case $\alpha<\frac12$. We present complete proofs of the cases $\alpha\ge\frac12$ and simplify some arguments along the way.
	We complement the above result by discussing the dynamical aspects of the model. We give a simple argument showing that, despite of the spectral transition, transport occurs for all energies for $\alpha=\frac12$. Finally, we discuss a theorem of Simon on dynamical localization in the sub-critical region $\alpha<\frac12$. This implies, in particular, that the spectrum is pure point in this regime.
\end{abstract}

\tableofcontents


\section{Introduction}


Disordered systems in material sciences have been the source of a plethora of interesting phenomena and many practical applications. The addition of impurities in otherwise fairly homogeneous materials is known to induce new behaviours such as Anderson localization where wave packets get trapped by the disorder and conductivity can be suppressed. 
Even though localization has been fairly studied for the one-dimensional Anderson tight-binding model \cite{Car,Car,CKM,DG,GMP,GZ,JZ,KS}, the multidimensional case is still a challenge where transition \textit{localization-delocalization} is expected for small disorder.

It is natural to expect that accurate mathematical models for disorder media should display an interesting phase diagram. 
In the case of Anderson localization, we consider the random operator
\begin{equation*}
	H = \Delta + \lambda V\quad \mathrm{on}\quad \ell^2(\Z^d).
\end{equation*}  
where $\Delta$ is the usual discrete Laplacian, $V$ acts by multiplication by a family of independent random variables and $\lambda$ is a coupling constant tuning the strength of the disorder. This is known as the Anderson model and can be seen as a simple model for the dynamics of an electron on a disordered lattice. In the absence of the potential, it is known that the spectrum of $H$ is absolutely continuous and wave packets are diffusive or `delocalized'. It is expected that, for small coupling constants, the effect of the disorder does not destroy conductivity and the system remains delocalized. In contrast, large coupling constants should lead to absence of diffusion or `localization'. Mathematically, this would be reflected at the spectral level by a transition from absolutely continuous to pure point spectrum.

The picture is completely understood in one-dimension: there is localization for all non-zero values of the coupling constant \cite{Car,GMP,KS,DSS}. In higher dimensions, localization has been proved for several regimes including high disorder, extreme energies and edge of the spectrum
(see \cite{AW} and references therein). However, there is no proof of absolutely continuous spectrum so far. 
Some results have been known on the Bethe lattice and tree graphs \cite{Kl,ASW,FHS}. A delocalization-localization transition has been proved for random Landau Hamiltonians where non-trivial transport occurs near Landau levels \cite{GKS}.

In order to understand how absolutely continuous spectrum survives in spite of the disorder, it has been proposed to modulate the random potential by a decaying envelope \cite{DSS, D, FGKM,KKO,Kr,KU}, this is, to replace $V(n)$ by $a_n V(n)$, where $(a_n)_n$ is a deterministic sequence satisfying $a_n \sim |n|^{-\alpha}$ for some decay rate $\alpha>0$.
For large values of $\alpha$ and $d\geq 3$, scattering methods can be applied, leading to the proof of absolutely continuous spectrum \cite{Kr}. Purely absolutely continuous spectrum was showed in \cite{JL}. A wider range of values of $\alpha$ was considered by Bourgain in dimension $2$\cite{B1} and higher \cite{B2}. Point spectrum was also showed to hold outside the essential spectrum of the operator in \cite{KKO}. 

In one dimension, transfer matrix analysis can be applied, leading to a complete understanding of the spectrum of the operator \cite{DSS, KLS}. This time however, the phase diagram of the model remains non-trivial and absolutely continuous spectrum can still be observed for large values of $\alpha$. As it is natural to expect, small values of $\alpha$ lead to pure point spectrum. Interestingly, there is a critical value of $\alpha$ for which a transition from pure point to singular continuous spectrum is observed as a function of the coupling constant. The three above regimes correspond to $\alpha>\frac12$, $\alpha<\frac12$ and $\alpha=\frac12$ respectively. A complete study of the spectral behaviour of the one-dimensional model is given in \cite{KLS}.

From the dynamical point of view, it is standard to show that the system propagates $\alpha>\frac12$. For the critical case $\alpha=\frac12$, no transition occurs at the dynamical level, despite of the spectral transition: there are non-trivial transport exponents for all values of the coupling constant \cite{GKT}. This provides yet another example of a model where spectral localization and transport coexist. Dynamical localization in the regime $0<\alpha<\frac12$ was showed in \cite{Si82} providing, as a by-product, a proof of point spectrum in this regime. In \cite{BMT02}, we provided a different proof of this behaviour for a continuum version of the model which can be easily adapted to the discrete case.

\subsection*{Structure of the article}

We introduce the model and state the main results in Section \ref{sec:results}. The transfer matrix analysis of \cite{KLS} is presented in details in Section \ref{sec:transfer-matrices}, some technical estimates being deferred to the Appendix. We show absolutely continuous spectrum for the super-critical case in Section \ref{sec:super-critical-case}. In Section \ref{sec:critical-case}, we discuss the spectral transition and the absence of dynamical localization in the critical case. Finally, we outline the proof of dynamical localization for the sub-critical case in Section \ref{sec:dyn-loc}.

\section{Model and Results}\label{sec:results}


We consider a one dimensional Schr\"odinger random operator $H_{\omega,\lambda}$ defined by 
\begin{equation}\label{op H}
 H_{\omega,\lambda} =\Delta +\lambda V_\omega\quad\mathrm{on}\quad \ell^2(\n),
\end{equation}
where $(\Delta x)(n)=x_{n+1}+x_{n-1}$ is the discrete Laplacian, $V_{\omega}$ is a random potential described below and $\lambda\in\R$ is a coupling parameter. Denoting by $\delta_n$ the $n$-th canonical vector, the action of the random potential $V_\omega$ is given by 
\begin{equation}\label{pot V}
V_\omega \delta_n = a_n \omega_n \delta_n, \quad n\in \n,
\end{equation}
where $(\omega_n)_{n\geq 0}$ are  i.i.d. bounded centered random variables defined on a probability space $(\Omega,\mathcal{F},\p)$ with bounded density $\rho$. We will assume that all $\omega_n$'s have variance equal to 1. The envelope on the environment is given by a positive sequence $a_n$ such that $\displaystyle\lim_{n\to\infty}n^{\alpha}a_n=1$ where $\alpha>0$ denote the decay rate.
Under these hypotheses $(H_{\omega})_{\omega}$ is non-ergodic family of self-adjoint operators.
Note that the random variables are assumed to be bounded so that the density $\rho$ has a bounded support. 


\subsection{Spectral Transition}


We recall the spectral results of \cite{KLS} which give the characterization of the spectrum of $H_{\omega,\lambda}$ in terms of the parameters.

\begin{theorem}\label{thm:main}
 Under the hypothesis above, the essential spectrum of $H_{\omega,\lambda}$ is $\p$-a.s. equal to $[-2,2]$. Furthermore, 

\begin{enumerate}
	\item[(1)] \textbf{Super-critical case.} If $\alpha>\frac12$ then for all $\lambda\in\mathbb{R}$, the spectrum of $H_{\omega,\lambda}$ is almost surely purely absolutely continuous in $(-2,2)$ 
					\vspace{1ex}

    \item[(2)] \textbf{Critical case.} If $\alpha=\frac12$ then for all $\lambda \neq 0$, the a.c. spectrum of $H_{\omega,\lambda}$ is almost surely empty. Moreover,
			\begin{enumerate}[label=\alph*.-]
			\item[a.] If $|\lambda|\geq 2$, the spectrum of $H_{\omega,\lambda}$ is almost surely pure point in $(-2,2)$.
			
			\vspace{1ex}
			
			\item[b.] If $|\lambda|< 2$, the spectrum of $H_{\omega,\lambda}$ is purely singular continuous in $\{|E|<\sqrt{4-\lambda^2}\}$ and pure point in $\{\sqrt{4-\lambda^2}<|E|<2\}$, almost surely.
		\end{enumerate}

	\vspace{1ex}
		
	\item[(3)] \textbf{Sub-critical case.} If $\alpha<\frac12$ then for all $\lambda\neq 0$, the spectrum of $H_{\omega,\lambda}$ is almost surely pure point in $(-2,2)$.
\end{enumerate}
\end{theorem}
We will provide complete proofs of Parts 1 and 2 in Sections \ref{sec:super-critical-case} and \ref{sec:critical-case} respectively. Part 3 is a consequence of the dynamical localization result discussed below.

\subsection{Dynamical Behavior}

	
	
	

Delocalization, or spreading of wave packets, for $\alpha>\frac12$ follows from the RAGE theorem. 
The situation is particularly interesting for $\alpha=\frac12$: non-trivial transport occurs regardless of the precise nature of the spectrum. In particular, this provides an example of an operator displaying pure point spectrum but no dynamical localization. 

To describe the dynamics, we consider the random moment of order $p\ge0$ at time $t$ for the time evolution, initially spatially localized on the origin and localized in energy by a positive function $f\in C_{0}^\infty(\R)$,
\begin{equation}
	\mathbb{M}_{\omega}(p,f,T):=\frac{2}{T}\int^{\infty}_0 \e^{-\frac{2t}{T}} \left\| |X|^{\frac{p}{2}} \e^{-itH_{\omega,\lambda}}f(H_{\omega,\lambda}) \delta_0 \right\|^2 dt,
\end{equation}
where $|X|$ denotes the position operator and $C^{\infty}_0(\R)$ is the space of infinitely differentiable compactly supported functions. 
The following result is proved in \cite{GKT} and 
establishes the absence of dynamical transition in the critical case.
	
\begin{theorem}\label{thm:GKT}
	Let $\alpha=\frac12$ and $\lambda \in \R$. The following holds $\p$-amost surely: for all positive $f\in C^{\infty}_0(\R)$ constantly equal to $1$ on a compact interval $J\subset (-2,2)$, for any $\nu>0$ and all $p>2\gamma_J +\nu$ where $\gamma_J=\inf\{\lambda(8-2E^2)^{-1}:\, E\in J\}$, there exists $C_{\omega}(p,J,\nu)$ such that
	\begin{eqnarray}
		\mathbb{M}_{\omega}(p,f,T) \geq C_{\omega}(p,J,\nu) T^{p-2\gamma_J-\nu}.
	\end{eqnarray}
\end{theorem}
%

In Proposition \ref{thm:critical-delocalization}, we establish a weaker result that already implies the absence of dynamical localization in the critical regime where $\alpha=\frac12$. 
\bigskip

To characterize the dynamical localization, we define the eigenfunction correlator 
\begin{equation}\label{correlator}
	Q_\omega(m,n;I)=\sup_{\substack{f\in C_0(I)\\ \norm{f}_\infty\le1}}
	 \left| \langle \delta_m , P_I(H_{\omega,\lambda})f(H_{\omega,\lambda}) \delta_n \rangle \right|,
\end{equation}
where $P_I(H_{\omega})$ denotes the spectral projection of $H_{\omega,\lambda}$ on the interval $I$ and $C_0(I)$ is the space of bounded measurable compactly supported functions in $I$.
 We say that $H_{\omega,\lambda}$ exhibits dynamical localization in an interval $I\subset\R$ if we have
 \begin{equation}\label{DL}
  \sum_n\esp\left[Q_\omega(m,n;I)^2\right]<\infty,
 \end{equation}
for all $m\in\Z$.
We state the main result of \cite{Si82}.
\begin{theorem}\label{thm:dynamical-localization}
	Let $0<\alpha<\frac12$ and $\lambda\neq 0$. Then, for each $m\in\Z$ and each compact energy interval $I\subset (-2,2)$, there exist constants $C=C(m,I)>0$, $c(m,I)>0$ such that
	\begin{equation}\label{eq:FM-bound}
		\esp[Q_\omega(m,n;I)^2] \leq C e^{-cn^{1-2\alpha}},
	\end{equation}
	for all $m,\, n\in \Z$.
	 In particular, $H_{\omega,\lambda}$ exhibits dynamical localization in the interval $I$.
\end{theorem}
Although the lack of ergodicity of the model induces the dependence of \eqref{eq:FM-bound}  on the base site $m$, it can be showed that the bound \eqref{DL} still implies pure point spectrum and finiteness of the moments such that for all $p>0$,
\begin{equation*}
	\esp\left(\sup_{t\in\R}
		\norm{
			|X|^{\frac{p}{2}}\e^{-itH_{\omega}}P_I(H_{\omega})\varphi_0}^2
	\right)<\infty,
\end{equation*}
for all $\varphi_0 \in l^2(\Z)$ with bounded support. A proof of these simple facts can be found in \cite{BMT01} in a related discrete model or in \cite{BMT02} in the continuum.
%
\bigskip

The proof of Theorem \ref{thm:dynamical-localization} in \cite{Si82} uses the Kunz-Souillard method (KSM) \cite{KS,DG}. 
In \cite{BMT01}, we prove localization for the one-dimensional Dirac operator in a sub-critical potential by means of the fractional moments method instead. This approach allowed us to greatly generalize the hypothesis required in \cite{Si82}. For instance, we can handle unbounded potentials under mild regularity assumptions on their law, a context which is out of the scope of the KSM. Our approach, which is discussed in Section \ref{sec:dyn-loc} below, can be easily adapted to the Anderson model providing an alternative proof of Theorem \ref{thm:dynamical-localization} under more general assumptions.
Furtheremore, in \cite{BMT02}, we proved the corresponding result for a continuum version of the model, a problem which was left open in \cite{DS} where the authors develop a continuum version of the KSM.

The analysis in \cite{BMT02} provides a control on the eigenfunctions of $H_{\omega,\lambda}$. Let $x_{\omega,E}=(x_{\omega,E,n})_n$ denote the eigenfunction of $H_{\omega}$ corresponding to the eigenvalue $E$. In \cite[Theorem 8.6]{KLS}, it is showed that
\begin{equation*}
	\lim_{n\to \infty} \frac{1}{n^{1-2\alpha}} \log \sqrt{|x_{\omega,E,n}|^2 + |x_{\omega,E,n-1}|^2} = -\frac{(1-2\alpha) \lambda^2}{2(4-E^2)}, \quad \p-\text{a.s.},
\end{equation*}
for almost every fixed $E\in(-2,2)$. In particular, this shows that for almost every $E\in(-2,2)$, $\p$-almost surely, there exists a constant $C_{\omega,E}$ such that
\begin{equation*}
	\sqrt{|x_{\omega,E,n}|^2 + |x_{\omega,E,n-1}|^2} \leq C_{\omega,E}\ \e^{-\frac{(1-2\alpha) \lambda^2}{2(4-E^2)}n^{1-2\alpha}}.
\end{equation*}
It is known that certain types of decay of eigenfunctions are closely related to dynamical localization \cite{DeRJLS1,DeRJLS2,GT}. Such criteria usually require a control on the localization centres of the eigenfunctions, uniformly in energy intervals. This information is missing in the above bound. In \cite{BMT02} (see also \cite{BMT01}), it is showed that:
\begin{proposition} \label{thm:bounds-eigenfunctions}
	Let $0<\alpha<\frac12$.
	For all compact energy interval $I\subset (-2,2)$, there exists two deterministic constants $c_1=c_1(I),\, c_2=c_2(I)$ and positive random quantities $c_{\omega}=c_{\omega}(I),\, C_{\omega}=C_{\omega}(I)$ such that
	\begin{eqnarray}\label{eq:SULE}
		c_{\omega} \, e^{-c_1 n^{1-2\alpha}} \leq \sqrt{|x_{\omega,E,n}|^2 + |x_{\omega,E,n-1}|^2} \leq C_{\omega} e^{-c_2 n^{1-2\alpha}}, \quad \p-\text{a.s.},
	\end{eqnarray}
	for all $E\in I\cap\sigma(H_{\omega,\lambda})$ and all $n\in\Z$.
\end{proposition}
 This asymptotics, although less precise on the exact rate of decay, is uniform in energy intervals. The upper bound \eqref{eq:SULE} can be seen as a form of the condition SULE  where the localization centres are all equal to $0$ (see \cite{DeRJLS1,GK1}, equation (2)).
 
 As a corollary of the lower bound above, we can show that certain stretched exponential moments diverge, a fact that characterizes in some sense the strength of the localization.
\begin{proposition}\label{thm:lower-bounds-DL}
	Let $0<\alpha<\frac12$ and $\lambda>0$.
	Let $I \subset (-2,2)$ be a compact interval contained in the interior of $\sigma_{pp}(H_{\omega,\lambda})$. Then, 
	for $\p$-almost every $\omega$, 
	\begin{eqnarray}
		\limsup_{t\to \infty} \left\| \exp\{|X|^{p}\}\ \e^{-it H_{\omega,\lambda}} \psi\right\|^2 = \infty,
	\end{eqnarray}
	for all $p>1-2\alpha$ and $\psi \in {\text Ran}P_I(H_{\omega,\lambda})$.
\end{proposition}

 
\section{Transfer Matrix Analysis}\label{sec:transfer-matrices}

Let $|E|\leq 2$ and consider the eigenvalue equation
\begin{eqnarray}
	x_{n+1} + x_{n-1} + V_{\omega,n} x_n&=& E x_n, \quad n\geq 1,
\end{eqnarray}
with some initial condition $x_0=a,\, x_1=b$.
The formal solution of this equation can be obtained via transfer matrices. Let 
$\displaystyle X_n = 
\begin{pmatrix}
	x_n \\ x_{n-1}
\end{pmatrix}$
for $n\geq 1$ where
$\displaystyle X_1 = 
\begin{pmatrix}
	b \\ a
\end{pmatrix}$. Then $X_{n+1}=T_{\omega,n} X_n$ where
\begin{eqnarray}
	T_{\omega,n} = T_{\omega,n}(E) = 
	\begin{pmatrix}
		E- V_{\omega,n} & -1\\
		1 & 0	
	\end{pmatrix}.
\end{eqnarray}
Iterating this relation yields 
$$X_n = T_{\omega,n-1} \cdots T_{\omega,2} T_{\omega,1} X_1.$$

The free system
\begin{eqnarray*}
	x_{n+1} + x_{n-1}&=& E x_n, \quad n\geq 1\\
	x_1 &=& Ex_0,
\end{eqnarray*}
with $x_0=a,\, x_1=b$,
has a basis of solutions given by $\varphi^+_n= \cos (kn)$ and $\varphi^-_n = \sin (kn)$
where $2\cos k = E$, which correspond to initial conditions 
$\begin{pmatrix}1\\ \cos k\end{pmatrix}$ and $\begin{pmatrix}0\\ \sin k\end{pmatrix}$
respectively. Once more, the general solutions 
$\displaystyle \Phi_n
=
\begin{pmatrix}
	\varphi_n \\ \varphi_{n-1}
\end{pmatrix}
$
can be given in terms of transfer matrices so that $\Phi_n=T^{n}\Phi_0$ where
\begin{equation*}
	T= 
	\begin{pmatrix}
		E& -1\\
		1 & 0	
	\end{pmatrix},
	\quad
	\Phi_0=
	\begin{pmatrix}
		a \\ 0
	\end{pmatrix}.
\end{equation*}

In the case of decaying potentials, $T$ corresponds to the asymptotic transfer matrix for vanishing potentials.
 
 We seek for a new basis in which the transfer matrix of the perturbed system is a perturbation of the identity. Hence, it is natural to express the system in the basis given by the free solutions. Let
 \begin{eqnarray}\label{eq:change-basis}
 	\Psi_n =
 	\begin{pmatrix}
 		\cos(kn) & \sin(kn)\\
 		\cos(k(n-1)) & \sin(k(n-1))
 	\end{pmatrix}\quad \mathrm{for}
 	\quad n\geq 1,
 \end{eqnarray}
 and define new coordinates $(Y_n)_n$ such that $X_n = \Psi_n Y_n$. This is usually known as the Pr\"ufer transform. The recurrence relation for this new representation is given by $Y_{n+1}= M_{\omega,n} Y_n$ where $M_{\omega,n} = \Psi_{n+1}^{-1} T_{\omega,a} \Psi_n$. Noting that $\Psi_{n+1}=T\Psi_n$, we factorize this in a convenient way as
\begin{eqnarray*}
	M_{\omega,n} = \Psi_{n+1}^{-1} (T_{\omega,n} T^{-1}) T \Psi_n = \Psi_{n+1}^{-1} (T_{\omega,n} T^{-1}) \Psi_{n+1}.
\end{eqnarray*}
A simple computation shows that $\displaystyle M_{\omega,n} = I + \frac{V_{\omega,n} }{\sin k} A_n$ where
\begin{eqnarray*}
	A_n =
	\begin{pmatrix}
		\cos(nk) \sin(nk) & \sin^2(nk)\\
		-\cos^2(nk) & - \cos(nk) \sin(nk)
	\end{pmatrix}.
\end{eqnarray*}
Summarizing, our new recurrence corresponds to
\begin{eqnarray}\label{eq:main-rec-vectors}
	Y_{n+1} = \left( I + \frac{V_{\omega,n} }{\sin k}\ A_n\right) Y_n.
\end{eqnarray}
Write $\displaystyle
Y_n=
\begin{pmatrix}
	y_n \\ y_{n-1}
\end{pmatrix}
$
and define new sequences $(\rho_n)_n$, $(R_n)_n$ and $(\theta_n)_n$ through the relation
\begin{eqnarray*}
	\rho_n = y_n + iy_{n-1} = R_n \e^{i\theta_n} \quad \mathrm{for}\quad  n\geq 1,
\end{eqnarray*}
where we adopt the convention $\theta_1 \in [0,2\pi)$ and $\theta_n-\theta_{n-1}\in (0,\pi)$ for $n\geq 2$. Note that all the above quantities depend on the disorder $\omega$ which we will hide from the notation. Note also that the original variables can be easily recovered from the Pr\"ufer variables:
\begin{eqnarray}\label{eq:prufer-developped}
	\begin{pmatrix}
		x_{n} \\ x_{n-1}
	\end{pmatrix}	 
	= R_n
	\begin{pmatrix}
		\cos (\bar{\theta}_n) \\
		\cos( \bar{\theta}_n-k)
	\end{pmatrix} \quad \mathrm{with}\quad \bar{\theta}_n = nk-\theta_n.
\end{eqnarray}
\begin{remark}
Let $\mathcal{F}_n$ be the $\sigma$-algebra generated by the random variables $(\omega_k)_{0\le k\le n} $ so that $\mathcal{F}_n=\sigma(\omega_0,\cdots,\omega_n)$.
We note that $T_{\omega,j} \in \mathcal{F}_n$ for $j\in\{0,\cdots, n\}$ and that $T_{\omega,j}$ is independent of $\mathcal{F}_n$ for $j>n$. In particular, $X_{n},\, Y_{n},\, \rho_{n},\, R_{n},\, \theta_n \in \mathcal{F}_{n-1}$ and are hence independent of $\omega_n$. Similar considerations hold for negative values of $n$. These simple facts will turn certain objects into martingales.
\end{remark}
We seek for a more convenient form of the recursion \eqref{eq:main-rec-vectors}
identifying $A_n$ as a linear transform on the complex plane such that
\begin{eqnarray*}
		A_n \e^{i\theta}=-i\cos(nk-\theta)\ \e^{ink}=-i\cos(nk-\theta)\ \e^{i(nk-\theta)}\ \e^{i\theta}.
\end{eqnarray*}
We can see that \eqref{eq:main-rec-vectors} takes the simple form
\begin{eqnarray}\label{eq:main-rec-complex}
	\rho_{n+1} = \left( 1 - i \frac{V_{\omega,n}}{\sin k} \cos(\bar{\theta}_n)\ \e^{i\bar\theta_n}\right) \rho_n,
\end{eqnarray}
where $\bar{\theta}_n = nk-\theta_n$.
The recurrence in terms of $(R_n, \theta_n )_n$ becomes
\begin{eqnarray}
	\label{eq:main-rec-with-phases}
	R_{n+1}^2 
	&=&
	R_{n} ^2 \left(
	1
	- \frac{ V_{\omega,n} }{\sin k} \sin (2 \bar \theta_n)
	+
	\frac{V_{\omega,n}^2}{\sin^2 k} \cos^2(\bar \theta_n)
	\right),
\end{eqnarray}
which can be iterated to yield
\begin{eqnarray}\label{eq:main-rec}
	R_{n+1}^2
	&=&
	R_1^2  \,
	\prod^n_{j=1}\left(
	1
	- \frac{ V_{\omega,j} }{\sin k} \sin (2\bar\theta_j)
	+
	\frac{ V_{\omega,j}^2}{\sin^2 k} \cos^2(\bar\theta_j)
	\right),
\end{eqnarray}
where here and below we always assume $R_1=1$.
One can show that the phases $(\bar\theta_n)_n$ satisfy the recursion
\begin{eqnarray}\label{eq:main-rec-phases}
	\tan(\bar{\theta}_{n+1}-k) = \tan(\bar{\theta}_{n}) + \frac{V_{\omega,n}}{\sin k}.
\end{eqnarray}
The recursion \eqref{eq:main-rec-with-phases} will be the starting point of our analysis but first, we have to show that the asymptotics of the $(X_n)_n$ and $(Y_n)_n$ systems are equivalent. We start noticing that
\begin{eqnarray*}
	\text{Tr} (\Psi^*_n \Psi_n) \leq 4,\quad \quad \text{det} (\Psi^*_n \Psi_n) = \sin^2 k.
\end{eqnarray*}
Hence, if $0<\lambda_1 < \lambda_2$ are the eigenvalues of $\Psi^*_n \Psi_n$ then
\begin{eqnarray*}
	\lambda_1+\lambda_2 \leq 4,\quad \quad \lambda_1 \lambda_2 = \sin^2 k,
\end{eqnarray*}
and thus
\begin{eqnarray*}
	\lambda_1 = \frac{\sin^2 k}{\lambda_2} = \frac{\sin^2 k}{\text{Tr} (\Psi^*_n \Psi_n)-\lambda_1} \geq \frac{\sin^2 k}{4}.
\end{eqnarray*}
Since $\| \Psi_n \|^2 \leq 4$, we obtain
\begin{eqnarray}
	\frac{\sin^2 k}{4}R_n^2 \leq \| X_n \|^2 \leq 4 R_n^2,
\end{eqnarray}
where we recall that $R_n=\| Y_n \|$.
It turns out that the variables $( R_n)_n$ also allow us to control the asymptotics of the transfer matrices. Let $\mathbf{T}_{\omega,n} = T_{\omega,n} \cdots T_{\omega,2} T_{\omega,1}$ so that $X_{n+1} = \mathbf{T}_{\omega,n} X_1$. 
We write $Y_n(\theta)$ when the recursion \eqref{eq:main-rec-vectors} is started from $\widehat Y_1 = \widehat \theta := \begin{pmatrix}\cos \theta\\ \sin \theta\end{pmatrix}$, and similarly for $R_n(\theta)$.
\begin{lemma} \label{thm:comparison}
For any pair of initial angles $\theta_1 \neq \theta_2$ and initial norm $\| Y_1\|=1$, there exists constants $c(\theta_1,\theta_2),\, C(\theta_1,\theta_2)$ such that
\begin{eqnarray}\nonumber
	c(\theta_1,\theta_2) \max\{ R_n(\theta_1),\, R_n(\theta_2)\}
	\leq
	\| \mathbf{T}_{\omega,n-1} \|
	\leq	
	C(\theta_1,\theta_2) \max\{ R_n(\theta_1),\, R_n(\theta_2)\},
\end{eqnarray}
for all $n\geq 1$.
\end{lemma}
\begin{proof}
	From the relation $\mathbf{T}_{\omega,n-1} X_1 = \Psi_n Y_n$, we obtain
	\begin{eqnarray}
		\| \mathbf T_{\omega,n-1} \|^2 \geq \| \Psi_n Y_n \|^2 \geq \frac{\sin^2 k}{4} \| Y_n\|^2.
	\end{eqnarray}
	The upper bound is more delicate and follows from a general result on unimodular matrices that we defer to Lemma \ref{thm:unimodular} in Appendix \ref{app:unimodular}.
\end{proof}

Starting from this point, we specialize to the random vanishing case. Recall that $(\omega_n)_{n\geq 0}$ is an i.i.d. sequence of bounded centered random variables defined on a probability space $(\Omega,\mathcal{F},\p)$ with an absolutely continuous distribution and denote by $\esp$ the expected value with respect to $\p$ so that $\esp(\omega_n)=0$ for all $n\geq 0$. We also assume that $\esp(\omega^2_n)=1$.

All the quantities defined above depend on $\omega$, $\lambda$ and $E$ but these parameters will be omitted from most of the notations whenever no confusion is possible.

 
\section{Super-critical Case: Absolutely Continuous Spectrum}\label{sec:super-critical-case}
This is based on a criterion of Last and Simon \cite{LS} that relates spectral properties to transfer matrices behavior. Let $\mathcal{T}_n(E)$ denote the product of transfer matrices associated to a bounded Schr\"odinger operator $H$ on $l^2([0,\infty))$ and an energy $E$.
\begin{theorem}\cite[Teorem 1.3]{LS}
	Suppose that
	\begin{eqnarray}\label{eq:LS-criterion}
		\liminf_n \int^b_a \| \mathcal{T}_n(E)\|^4dE <\infty.
	\end{eqnarray}
	Then, $(a,b)\subset \sigma(H)$ and the spectral measure is purely absolutely continuous on $(a,b)$.
\end{theorem}
 The criterion is valid for any power larger than 2. There is nothing special about the power 4 except that it makes the computations easier. 
 
 In the following, we write $\mathbf{T}_{\omega,n}(E)$, $Y_n(E)$ and $R_n(E)$ when we want to emphasise the dependence on the energy $E$. We write $V_{\omega,n}=\lambda a_n\omega_n$ for $n\geq 0$.
\begin{proof}[Proof of Theorem \ref{thm:main}, part 1 (super-critical case)]
Let $\theta_0$ be any initial angle and let $[a,b]\subset (-2,2)$. 
According to Lemma \ref{thm:comparison}, it is enough to show
\begin{eqnarray}\label{eq:bound-int-rho}
	\liminf_n \esp\left[  \int^b_a R_n^4(E)dE\right]<\infty.
\end{eqnarray}
Then, by Fatou's lemma,
\begin{eqnarray}
	\esp\left[ \liminf_n \int^b_a R_n^4(E)dE \right] \leq \liminf_n \esp\left[  \int^b_a R_n^4(E) dE\right]<\infty,
\end{eqnarray}
which implies that \eqref{eq:LS-criterion} holds almost surely. 
Squaring \eqref{eq:main-rec-with-phases}, we obtain
\begin{eqnarray}
	R_{n+1}^4 (E)= \left\{ 1 - 2\frac{V_{\omega,n}}{\sin k}\sin(2\bar{\theta}_n) + A_{\omega,n}(E) \right\} R_n^4(E),
\end{eqnarray}
where $A_{\omega,n}(E)$ collects all the terms of higher order in $V_{\omega,n}$. An inspection at those terms shows that there exists $c=c(a,b)\in(0,\infty)$ such that $|A_{\omega,n}(E)| \leq c n^{-2\alpha}$ for all $E\in [a,b]$.

Observing that $R_n$ is $\mathcal{F}_{n-1}$-measurable and bounded, we have
\begin{eqnarray}
	\esp\left[ R_{n+1}^4 (E) \Big{|} \mathcal{F}_{n-1} \right]
	= 
	\esp\left[ 
		1 - 2\frac{V_{\omega,n}}{\sin k}\sin(2\bar{\theta}_n) + A_{\omega,n}(E) 
	\Big{|} \mathcal{F}_{n-1} \right]
	R_n^4(E).
\end{eqnarray}
Now, as $\bar \theta_n$ is $\mathcal{F}_{n-1}$-measurable and $V_{\omega,n}$ is independent of $\mathcal{F}_{n-1}$ and centered,
\begin{eqnarray}
	\esp\left[V_{\omega,n}\sin(2\bar{\theta}_n) \Big{|} \mathcal{F}_{n-1} \right]
	=
	\sin(2\bar{\theta}_n) \esp\left[V_{\omega,n}\right]
	=
	0.
\end{eqnarray}
Collecting all the above observations, we conclude that
\begin{eqnarray}
	\esp\left[ R_{n+1}^4 (E) \Big{|} \mathcal{F}_{n-1} \right]
	\leq 
	\left\{ 1 + \frac{c}{n^{2\alpha}}\right\} \, R_n^4(E),
\end{eqnarray}
for all $E\in [a,b]$ and all $n\geq 1$. Integrating with respect to $\p$ and iterating, we obtain
\begin{eqnarray}
	\esp\left[ R^4_{n+1}(E) \right] 
	\leq 
	\left\{ 
		1 + \frac{c}{n^{2\alpha} }
	\right\} 
	\esp\left[ R^4_{n}(E) \right] 
	\leq 
	\prod^n_{j=1}
	\left\{ 
		1 + \frac{c}{j^{2\alpha} }
	\right\},
\end{eqnarray}
for all $E\in [a,b]$ and all $n\geq 1$. As $\displaystyle \sum_j j^{-2\alpha} < \infty$ for $\alpha>\tfrac12$, the product above is bounded uniformly in $n$ and $E\in[a,b]$.  This finishes the proof.
\end{proof}

\section{Critical Case: transition from Singular Continuous to Point Spectrum}\label{sec:critical-case}

\subsection{General scheme.}
To prove the absence of a.c. spectrum, we use the following criterion of Last and Simon. With the notations of the beginning of Section \ref{sec:super-critical-case}:
\begin{theorem}\cite[Theorem 1.2]{LS}\label{thm:criterion-no-ac}
	Suppose $\displaystyle\lim_{n\to\infty} \| \mathcal{T}_n(E)\| = \infty$ for a.e. $E\in[a,b]$. Then, $\mu_{ac}([a,b])=0$, where $\mu_{ac}$ is the absolutely continuous part of the spectral measure for $H$.
\end{theorem}
We will prove that, for all $\alpha\leq \frac12$, $\lambda\in\R$ and all $E\in(-2,2)$ corresponding to values of $k$ different from $\frac{\pi}{4}$, $\frac{\pi}{2}$ and $\frac{3\pi}{4}$, one has
\begin{eqnarray}
	\beta=\beta(E,\lambda):=\lim_{n\to\infty} \frac{\log \norm{\bT_{\omega,n}(E)}}{\sum^n_{j=1}j^{-2\alpha}} = \frac{\lambda^2}{2(4-E^2)},\quad \p-a.s.
\end{eqnarray}
which in particular implies that the norms diverge so that Theorem \ref{thm:criterion-no-ac} can be applied. This will be done by working the fine asymptotics of $(R_n)_n$.

Next, we will argue that, for all $\lambda$ and $E$ as above, there exists a random initial direction $\hat \vartheta_0$ such that
\begin{eqnarray}
	\lim_{n\to\infty} \frac{\log \|\bT_{\omega,n}(E) \hat\vartheta_0\|}{\log n} = -\beta,\quad \p-a.s.
\end{eqnarray}
This guaranties that the solution of the eigenvalue equation with initial conditions $x_0=\cos \vartheta_0$ and $x_1=\sin \vartheta_0$ satisfies
\begin{eqnarray}
		\lim_{n\to\infty} \frac{\log \sqrt{x^2_n+x^2_{n-1}}}{\log n} = -\beta,\quad \p-a.s.
\end{eqnarray}
Summarising, for almost every pair $(E,\omega)$ there is a unique decaying solution with rate of decay $n^{-\beta}$.
If $\beta>\frac12$, it is in $ l^2(\n)$. This holds if and only if $|\lambda|>\sqrt{4-E^2}$. This condition is always met if $|\lambda|\geq 2$. If $|\lambda|<2$, the condition fails once $|\lambda|\leq\sqrt{4-E^2}$  or, equivalently, if $|E|\leq\sqrt{4-\lambda^2}$. In this case, there is no $l^2$ solution. The result follows from the general theory of rank 1 perturbations (see, for instance, Section II in Simon's lecture notes \cite[Section II]{Si95}).
 
\subsection{Asymptotics of Transfer Matrices}

The following computation is valid in the critical and sub-critical region i.e. $\alpha\leq \frac12$. Recall the relation $E = 2\cos k$.
\begin{proposition}\label{thm:lyapuvov}
	Let $0<\alpha\leq\frac12$. Let $R_n(E,\theta_0)$ denote the recurrence \eqref{eq:main-rec-with-phases} corresponding to an energy $E$, an initial direction $\theta_0$ and $R_1=1$. Then, for all initial direction $\theta_0$ and $E\in (-2,2)$ corresponding to values of $k$ different from $\frac{\pi}{4}$, $\frac{\pi}{2}$ and $\frac{3\pi}{4}$, we have
	\begin{eqnarray}
		\lim_{n\to\infty} \frac{\log R_n(E,\theta_0)}{\sum^n_{j=1}j^{-2\alpha}} = \frac{\lambda^2}{8 \sin^2 k},\quad \p-a.s.
	\end{eqnarray}
\end{proposition}
\begin{proof}
	We take $a_0=1$ and $a_n=n^{-\alpha}$ to simplify the presentation and
	write again $V_{\omega,n}=\lambda n^{-\alpha}\omega_n$ for $n\geq 0$.
	Remember the recursion,
	\begin{eqnarray}
	R_n^2 
	&=&
	\prod^{n-1}_{j=1}\left\{ 
	1
	- \frac{V_{\omega,j }}{\sin k} \sin (2\bar\theta_j)
	+
	\frac{V_{\omega,j}^2}{\sin^2 k} \cos^2(\bar\theta_j)
	\right\}.
\end{eqnarray}
Using the Taylor expansion $\log(1+\varepsilon)= \varepsilon -\frac{\varepsilon^2}{2}+O(\varepsilon^3)$, we obtain
\begin{eqnarray}
	&&\log  \prod^{n-1}_{j=1}\left\{ 
	1
	- \frac{V_{\omega,j} }{\sin k} \sin (2\bar\theta_j)
	+
	\frac{V_{\omega,j}^2}{\sin^2 k} \cos^2(\bar\theta_j)
	\right\}\\
	&=&
	\sum^{n-1}_{j=1}
	\left\{ 
	- \frac{V_{\omega,j} }{\sin k} \sin (2\bar\theta_j)
	+
	\frac{V_{\omega,j}^2}{\sin^2 k} \left(\cos^2(\bar\theta_j)-\frac12 \sin^2 (2\bar\theta_j)\right)
	+
	O(V^3_{\omega,j})
	\right\}\\
	&=&
	\sum^{n-1}_{j=1}
	\left\{ 
	\frac{\esp[V_{\omega,j}^2]}{4\sin^2 k}
	- \frac{V_{\omega,j} }{\sin k} \sin (2\bar\theta_j)
	+
	\frac{V_{\omega,j}^2-\esp[V^2_{\omega,j}]}{4\sin^2 k} \right.\\
	&&
	\phantom{lospollitos}
	+
	\frac{V_{\omega,j}^2-\esp[V_{\omega,j}^2]}{\sin^2 k}\left( \frac{1}{2}\cos 2 \bar \theta_j + \frac{1}{4} \cos 4 \bar \theta_j\right)\\
	&&\label{eq:martingale-decomposition}
	\left.
	\phantom{lospollitos}
	+
	\frac{\esp[ V_{\omega,j}^2]}{\sin^2 k}\left( \frac{1}{2}\cos 2 \bar \theta_j + \frac{1}{4} \cos 4 \bar \theta_j\right)
	+
	O(V^3_{\omega,j})
	\right\}.
\end{eqnarray}
The main contribution comes from the first term above. The tree next terms are martingales with respect to $\mathcal{F}_n$ and are shown to be $\displaystyle o\left(\sum^{n-1}_{j=1} j^{-2\alpha} \right)$ using martingale arguments (Lemma \ref{thm:martingales} and Remark \ref{rk:martingale}).
The fifth term is more delicate and will be treated by ad-hoc methods in Appendix \ref{app:phases}.
Finally, the contribution from the $O(V^3_{\omega,j})$ term is $\displaystyle o\left(\sum^{n-1}_{j=1} j^{-2\alpha} \right)$ as well.  
\end{proof}
The following lemma is the key to estimate the martingale terms. It corresponds to \cite[Lemma 8.4]{KLS} but this shorter proof is taken from \cite{BMT01}.
\begin{lemma}\label{thm:martingales}
	Let $(X_j)_j$ be i.i.d. random variables with $\esp[X_n]=0$, let $\mathcal{F}_n=\sigma(X_1,\cdots, \, X_n)$ and let $Y_n \in \mathcal{F}_{n-1}$ for $j\geq 1$.
	Let $\gamma>0$ and define
	\begin{eqnarray}
		M_n = \sum^n_{j=1} \frac{X_j Y_j}{j^{\gamma}} \quad \text{and} \quad s_n = \sum^n_{j=1} \frac{1}{j^{2\gamma}}.
	\end{eqnarray}
	Assume $|X_n|,\, |Y_n|\leq 1$. Then, $(M_n)_n$ is an $(\mathcal{F}_n)_n$-martingale and
	\begin{enumerate}[label=\alph*.-]
		\item[(i)] For $\gamma \leq \frac12$ and all $\varepsilon>0$,
			\begin{eqnarray}
				\lim_{n\to\infty} s_n^{-\frac{1+\varepsilon}{2}}M_n= 0,\quad \quad \p-a.s.
			\end{eqnarray}
	
		\vspace{1ex}
		
		\item[(ii)] For $\gamma>\frac12$, $(M_n)_n$ converges $\p$-almost surely to a finite (random) limit $M_{\infty}$ and, for all $\kappa<\gamma - \frac12$, we have
		\begin{eqnarray}
			\lim_{n\to\infty} n^{\kappa} \left( M_{\infty}-M_n\right) = 0,\quad \p-a.s.
		\end{eqnarray}
	\end{enumerate}		
\end{lemma}
\begin{proof}
	The reader can consult the book \cite{Durrett} for the general properties of martingales used below. 
	The sequence $(M_n)_n$ is a martingale thanks to our hypothesis on $(X_n)_n$ and $(Y_n)_n$: indeed, since $M_n,\, Y_n \in \mathcal{F}_n$ and $X_{n+1}$ is independent of $\mathcal{F}_n$ and centered, we have
	\begin{eqnarray}
		\esp[M_{n+1} | \mathcal{F}_n]
		&=&
		\esp\left[\frac{X_{n+1}Y_{n+1}}{(n+1)^{\gamma}}  + M_{n} \Big{|}\mathcal{F}_n\right]\\
		&=&
		\esp[X_{n+1}]\frac{Y_{n+1}}{(n+1)^{\gamma}} + M_n = M_n.
\end{eqnarray}		
	Let $\gamma\leq \frac12$. We use Azuma's inequality \cite{Azuma}: let $(M_n)_n$ be a martingale such that $|M_n-M_{n-1}|\leq c_n$ for all $n\geq 1$. Then,
	\begin{eqnarray}
		\p\left[ |M_n-M_0| \geq t\right] \leq 2 \exp\left\{ - \frac{t^2}{2 \sum^n_{j=1} c_j^2}\right\}.
	\end{eqnarray}
	In our case, $M_0=0$, $c_j = 2 j^{-\gamma}$, and taking $t=s_n^{\frac{1+\varepsilon'}{2}}$ for $0<\varepsilon'<\varepsilon$, we obtain
	\begin{eqnarray}
		\p\left[ |M_n| \geq s_n^{\frac{1+\varepsilon'}{2}}\right] \leq 2 \exp\left\{ - C n^{\varepsilon'}\right\},
	\end{eqnarray}
	for some $C>0$. The claim $a.-$ then follows from Borel-Cantelli's lemma.
	
	Now, let $\gamma>\frac12$. Noticing that, for $k<l$,
	\begin{eqnarray}
		\esp[X_kY_k X_l Y_l] 
		= \esp[\esp[X_kY_k X_l Y_l| \mathcal{F}_{l-1}]]
		= \esp[X_kY_k X_l \esp[Y_l| \mathcal{F}_{l-1}]]
		=0,
	\end{eqnarray}
	we have
	\begin{eqnarray}
		\sup_n\esp[M_n^2] = \sup_n \sum^n_{j=1} \frac{\esp[X_j^2 Y_j^2]}{j^{2\gamma}} \leq \sum_{j\geq 1} \frac{1}{j^{2\gamma}}<\infty.
	\end{eqnarray}
	Hence, $(M_n)_n$ is an $L^2$-martingale and converges, i.e., there exists a random variable $M_{\infty}$ such that $\lim_{n\to\infty}M_n=M_{\infty}$, $\p$-a.s.. Finally, applying Azuma's inequality to the martingale $(M_{n+k}-M_n)_{k\geq0}$, we obtain
	\begin{eqnarray}
		\p\left[ n^{\kappa}\left| M_{n+k}-M_n\right| \geq 1\right]\leq 2 \exp\left\{ - c n^{2(\gamma-\frac12 -\kappa)}\right\},
	\end{eqnarray}
	for all $k\geq0$.
	Choosing $\kappa<\gamma-\frac12$, the last claim follows from Fatou's lemma, the convergence of $(M_n)_n$ and Borel-Cantelli.
\end{proof}
\begin{remark}\label{rk:martingale}
	
	The martingale property for the second, third and fourth terms in decomposition \eqref{eq:martingale-decomposition} follows from Lemma \ref{thm:martingales}, recalling that $\bar{\theta}_j \in \mathcal{F}_n$ for all $j\leq n+1$.
	
	All the martingale terms in Proposition \ref{thm:lyapuvov} are seen to be $o\left(\sum^n_{j=1} j^{-2\alpha }\right)$ by taking $\gamma=\alpha$ in Lemma \ref{thm:martingales} for the first one and $\beta=2\alpha$ for the others. The lemma actually guaranties that they are much smaller.
	
	There is still room in Azuma's inequality to allow the support of the random variables $(X_n)_n$ to grow with $n$. Unbounded random variables could be handled under some moment assumptions with some extra effort.
\end{remark}
 
\subsection{Control of Generalized Eigenfunctions}
The next proposition provides $l^2$ eigenfunctions in the proper region.
For an angle $\theta$, we write $\widehat\theta=(\cos \theta, \sin \theta)$.
\begin{proposition}\label{thm:decay-eigenfunctions}
	Let $\alpha=\frac12$ and $k\neq \frac{\pi}{4},\, \frac{2\pi}{4},\, \frac{3\pi}{4}$. Then, for $\p$-almost every $\omega$, there exists an initial angle $\vartheta_0 = \vartheta_0(\omega)$ such that
	\begin{eqnarray}
		\lim_{n\to\infty} \frac{\log \| \mathbf{T}_{\omega,n}(E) \widehat\vartheta_0\|}{\sum^n_{j=1}j^{-2\alpha}} = -\frac{\lambda^2}{8 \sin^2 k},\quad \p-a.s.
	\end{eqnarray}
\end{proposition}
The proof is given in details in Appendix \ref{app:unimodular}. We now provide a non-asymptotic lower bound on eigenfunctions that will be the key to the proof of absence of dynamical localization. The next lemma actually provides a lower bound on any non-trivial generalized solution of $H_{\omega,\lambda}x = Ex$ for $\alpha=\frac12$, $\lambda>0$, uniformly in $E$ ranging over compact intervals of $(-2,2)$. The proof is taken from \cite{BMT01} where it was developed in the context of the discrete Dirac model.
\begin{lemma}\label{thm:lower-bound-eigenfunctions}
	Let $\alpha=\frac12$ and fix $\lambda>0$. For $E\in (-2,2)$ define $x_{\omega,E}=(x_{\omega,E,n})_n$ as the solution of 
	$H_{\omega,\lambda}x = Ex$ with a possibly random initial condition $\widehat{\vartheta}_0$.
	Then, for each compact interval $I\subset (-2,2)$, there exists a deterministic constant $\kappa=\kappa(I)>0$ such that, for $\p$-almost every $\omega$, there exists $c_{\omega}=c_{\omega}(I)>0$ such that
	\begin{eqnarray}
		\sqrt{| x_{\omega,E,n}|^2 + | x_{\omega,E,n-1}|^2} \geq c_{\omega} n^{-\kappa}, \quad \forall \, n\in \Z.
	\end{eqnarray}	 
\end{lemma}
\begin{proof}
	We prove the bound for all $n\geq 1$, the opposite case being similar.
	We can reconstruct $x_{\omega,E}$ through the recurrence
	$\displaystyle
		\begin{pmatrix}
			x_{\omega,E,n} \\ x_{\omega,E,n-1}
	\end{pmatrix}			
	=
	{\bf T}_{\omega,n-1}
	\widehat{\vartheta}_0
	$.
	This implies in particular that
	\begin{eqnarray}
		\sqrt{| x_{\omega,E,n}|^2 + | x_{\omega,E,n-1}|^2} 
		\geq 
		\| {\bf T}_{\omega,n-1}\|^{-1}.
	\end{eqnarray}	
	Hence, using Lemma \ref{thm:comparison} with some $\vartheta_1 \neq \vartheta_2$, 
	\begin{eqnarray}
		\p\left[ \| X_{\omega,E,n}\|  \leq n^{-\kappa}\right]
		&\leq&
		\p\left[ \| {\bf T}_{\omega,n-1}(E) \| \geq n^{\kappa}\right]\\
		&\leq& 
		n^{-2\kappa} \esp\left[ \| {\bf T}_{\omega,n-1} (E)\|^2\right]\\
		&\leq& 
		C_1(\vartheta_1,\vartheta_2) n^{-2\kappa}
		\left\{
			 \esp\left[ R_n^2(E,\vartheta_1)\right]+\esp\left[ R_n^2(E,\vartheta_1)\right]
		\right\},
	\end{eqnarray}
	for some $C_1(\vartheta_1,\vartheta_2)>0$.
	Keeping in mind the recursion \eqref{eq:main-rec}, the argument of the proof of Part 1 of Theorem \ref{thm:main} given in Section \ref{sec:super-critical-case} can be reproduced and yields
	\begin{eqnarray}
		\esp\left[ R_n^2(E,\vartheta_1)\right]
		&\leq& 
		\prod^n_{j=1} \left( 1 + \frac{b}{j} \right)
		\leq
		C_2 n^{b-2\kappa},
	\end{eqnarray}
	for some constants $b=b(I)$ and $C_2=C_2(b)$. The estimate for $R_n^2(E,\vartheta_2)$ is of course similar. Taking $b-2\kappa<-1$, the result follows by Borel-Cantelli.
%
%
\end{proof}


\subsection{Absence of dynamical localization}

The following result shows that no dynamical localization can arise in the critical case. This was already stated in Theorem \ref{thm:GKT} from \cite{GKT}. The simple argument below uses our lower bound on eigenfunctions from Lemma \ref{thm:lower-bound-eigenfunctions} and is once again taken from \cite{BMT01}.

\begin{proposition}\label{thm:critical-delocalization}
	Let $\alpha=\frac12$ and $\lambda>0$.
	Let $I$ be a compact interval contained in the interior of $\sigma_{pp}(H_{\omega,\lambda})$. Then, there exists $p_0>0$ such that,
	for $\p$-almost every $\omega$, 
	\begin{eqnarray}
		\limsup_{t\to \infty} \left\| |X|^{p/2} \e^{-it H_{\omega,\lambda}} \psi\right\|^2 = \infty,
	\end{eqnarray}
	for all $p>p_0$ and $\psi \in {\text Ran}P_I(H_{\omega,\lambda})$.
\end{proposition}
\begin{proof}
	Let $c_{\omega}$ and $\kappa$ be as in Lemma \ref{thm:lower-bound-eigenfunctions}. Let $(\psi_l)_{l}$ be a basis of ${\text Ran}P_I(H_{\omega,\lambda})$ consisting of eigenfunctions of the operator $H_{\omega,\lambda}$, with corresponding eigenvalues $(E_l)_l$.
	 Define the truncated position operator $X_N = X \chi_{[-N,N]}$. Then, taking $p>\kappa-1$,
	\begin{eqnarray}
		\left\| |X_N|^{p/2} \psi_l\right\|^2
		&=&
		\sum_{|n|\leq N} |n|^p |\psi_l(n)|^2\\
		&\geq&
		\frac12 \sum_{|n|\leq N-1} (|n|-1)^p \max\{|\psi_l(n)|^2,|\psi_l(n-1)|^2\}\\
		&\geq&
		\frac{c_{\omega}}{2} \sum_{|n|\leq N-1} (|n|-1)^p |n|^{-\kappa}
		\geq
		c'_{\omega} N^{p-\kappa+1},
	\end{eqnarray}		 
	for some $c'_{\omega}>0$ and for all $l$. Let $\psi \in {\text Ran}P_I(H_{\omega,\lambda})$ and write $\psi = \sum_l a_l \psi_l$. Hence,
	\begin{eqnarray}
		\left\| |X_N|^{p/2} e^{-i t H_{\omega,\lambda}}\psi\right\|^2
		=
		\sum_{l,l'} \alpha_l \overline{\alpha}_{l'}e^{-it(E_l - E_{l'})}
		\langle \psi_{l'}, |X_N|^{p/2} \psi_l \rangle.
	\end{eqnarray}
	After a careful application of the dominated convergence theorem to exchange sums and integrals, we obtain
	\begin{eqnarray}
		\lim_{T\to\infty} \frac{1}{T} \int^T_0 \left\| |X_N|^{p/2} \e^{-i t H_{\omega,\lambda}}\psi\right\|^2 dt
		=
		\sum_l |a_l|^2 \left\| |X_N|^{p/2} \psi_l\right\|^2
		\geq 
		c'_{\omega} N^{p-\kappa+1}.
	\end{eqnarray}
	Hence, there exists an diverging (random) sequence $(T_N)_N$ such that
	\begin{eqnarray}
		\frac{1}{T_N} \int^{T_N}_0 \left\| |X_N|^{p/2} \e^{-i t H_{\omega,\lambda}}\psi\right\|^2 dt
		\geq 
		\frac{c'_{\omega}}{2} N^{p-\kappa+1},
	\end{eqnarray}
	for all $N\geq 1$. From here, we can find a diverging (random) sequence $(t_N)_N$ such that
	\begin{eqnarray}
		\left\| |X_N|^{p/2} \e^{-i t_N H_{\omega,\lambda}}\psi\right\|^2
		\geq 
		\frac{c'_{\omega}}{4} N^{p-\kappa+1},
	\end{eqnarray}
	for all $N\geq 1$. This finishes the proof.
\end{proof}

\section{Sub-critical case: Dynamical Localization}\label{sec:dyn-loc}


Theorem \ref{thm:dynamical-localization} was proved in \cite{Si82} using the Kunz-Souillard method. We briefly outline the approach used in \cite{BMT02} to prove dynamical localization for the analogue continuum model and which becomes particularly simple when adapted to the discrete case (see \cite{BMT01} for the related discrete Dirac  model). We use the fractional moment method \cite{AM,AW}.

The key to this approach is to estimate the correlator by the fractional moments of the Green's function in boxes. Let $H_{\omega,L}$ denote the restriction of the operator $H_{\omega,\lambda}$ to $[-L,L]$, let $R_{\omega,L}(E)=(H_{\omega,L}-E)^{-1}$ and let $G_{\omega,L}(m,n)= \bra \delta_m, R_{\omega,L}(E) \delta_n \ket$.
A suitable adaptation of the arguments of \cite[Section 7]{AW} to the decaying potential case shows that
\begin{eqnarray}
	\esp\left[ Q_{\omega}(m,n;I )^2\right]
	\leq
	C
	|\lambda|^{-1}
	a_m^{-1}
	\liminf_{L\to\infty}
	\int_I \esp\left[ |G_{\omega,L}(m,n)|^s \right] \, dE,
\end{eqnarray}
for $s\in(0,1)$ small enough, for some $C=C(s)>0$ and all $m,n\in\n$. It is then possible to show that there exists a constant $C=C(s,I)$ such that
\begin{eqnarray}
	\esp\left[ |G_{\omega,L}(m,n)|^s \right] 
	\leq 
	C
	|\lambda|^{-2s} a_m^{-2s}
	\esp\left[ \| {\bf T}_{\omega,[m,n]}(E) \widehat{X}_m \|^{-s} \right],
\end{eqnarray}
for all $E\in I$ and all $m,n\in\n$, where ${\bf T}_{\omega,[m,n]}(E) = T_{\omega,n} \cdots T_{\omega,m}$ and $\widehat{X}_m = \|X_m\|^{-1}X_m$ with $X_m = {\bf T}_{\omega,m} X_0$. The inverse fractional moments of transfer matrices are shown to decrease exponentially in \cite{CKM} in the ergodic case. The proof uses the positivity of the Lyapunov exponent as an input. The decaying potential case requires a finer analyisis \cite{BMT02,BMT01} and uses the asymptotics of transfer matrices given in Proposition \ref{thm:lyapuvov} as a starting point. More precisely, it is possible to show that, for $s\in(0,1)$ small enough, there exists constants $C=C(m,s,I)>0$ and $c=c(s,I)>0$ such that
\begin{eqnarray}
	\esp\left[ \| {\bf T}_{\omega,[m,n]}(E) \widehat{X}_m \|^{-s} \right]
	\leq
	C e^{-cn^{1-2\alpha}},
\end{eqnarray}
for all $E\in I$ and all $n\in\n$. Theorem \ref{thm:dynamical-localization} then follows from a concatenation of the above estimates.

The proof of the upper bound in Proposition \ref{thm:bounds-eigenfunctions} in standard (see \cite[Theorem 9.22]{CFKS}). The lower bound can be proved along the lines of Lemma \ref{thm:lower-bound-eigenfunctions}. Proposition \ref{thm:lower-bounds-DL} can then be proved as Proposition \ref{thm:critical-delocalization} using Proposition \ref{thm:bounds-eigenfunctions} instead of Lemma \ref{thm:lower-bound-eigenfunctions}.

\appendix

\section{Some technical lemmas}


\subsection{Two results on unimodular matrices}\label{app:unimodular}
We say that a matrix is unimodular if it has determinant equal to $1$. The following lemma is used to compare the asymptotics of the transfer matrix with those of the sequence $(R_n)_n$.
\begin{lemma}\label{thm:unimodular}
	Let $A$ be an unimodular matrix and let $\hat \theta = (\cos \theta,\, \sin \theta)$. Then, for all pair of angles $|\theta_1-\theta_2|\leq \frac{\pi}{2}$,
	\begin{eqnarray}
		\| A \| \leq \sin \left( \tfrac{|\theta_1-\theta_2|}{2}\right)^{-1} \max \{ \| A \hat \theta_1\|,\, \| A\hat \theta_2\| \}.
	\end{eqnarray}
\end{lemma}
\begin{proof}
	First, there exists angles $\theta_0$ and $\sigma_0$ such that $A^* \hat \sigma_0 = \| A \| \hat \theta_0$. Then, for any angle $\theta$,
	\begin{eqnarray}\nonumber
		| \cos(\theta-\theta_0) | = |\langle \hat \theta, \hat \theta_0 \rangle |
		= \frac{1}{\| A \|} | \langle\hat \theta, A^*\hat \sigma_0 \rangle |
		= \frac{1}{\| A \|} | \langle A \hat \theta, \hat \sigma_0 \rangle |
		\leq \frac{\| A \hat \theta \|}{\| A \|}.
	\end{eqnarray}
	This way, for any pair of angles, we obtain
	\begin{eqnarray}\nonumber
		\| A \| \max\{ | \cos(\theta_1-\theta_0) |,\, | \cos(\theta_2-\theta_0) |  \} \leq \max\{ \| A\hat \theta_1\|,\, \| A\hat \theta_2\|\}.
	\end{eqnarray}
	To conclude, just note that for all $|\gamma|\leq \pi/2$, the minimum of the function $x\mapsto \max\{|\cos(x)|, $ $ |\cos(x+\gamma)|\}$ is attained at $\pi/2-\gamma/2$ and is equal to $|\cos(\pi/2-\gamma/2)|=|\sin(\gamma/2)|$.
\end{proof}
The following lemma is used to find eigenfunctions with the proper decay.
\begin{lemma}\label{thm:eigendirection}
	For a unimodular matrix with $\| A \| >1$, define $\vartheta=\vartheta(A)$ as the unique angle $\vartheta \in (-\frac{\pi}{2},\frac{\pi}{2}]$ such that $\| A \hat\vartheta \| = \| A \|^{-1}$. We also define $r(A)=\| A (1,0)^T\| / \| A (0,1)^T\|$.

	Let $(A_n)_n$ be a sequence of unimodular matrices with $\| A_n \| >1$ and write $\vartheta_n=\vartheta(A_)$ and $r_n = r(A_n)$.
	Assume that
	\begin{enumerate}[label=\alph*.-]
		\item[(i)] $\displaystyle \lim_{n\to \infty} \| A_n\| = \infty$,
		
		\vspace{1ex}
		
		\item[(ii)] $\displaystyle \lim_{n\to\infty} \frac{\| A_{n+1}A_n^{-1}\|}{\| A_n \| \, \| A_{n+1}\|}=0$.
	\end{enumerate}
	Then,
	\begin{enumerate}
		\item $(\vartheta_n)_n$ has a limit $\vartheta_{\infty}\in (-\pi/2,\pi/2)$ if and only if $(r_n)_n$ has a limit $r_{\infty}\in[0,\infty)$. If $\vartheta_n\to\pm\pi/2$, then $r_n\to\infty$ but, if $r_n\to\infty$, we can only conclude that $|\vartheta_n|\to\pi/2$.
		
		\vspace{1ex}
		
		\item Suppose $(\vartheta_n)_n$ has a limit $\vartheta_{\infty}\neq 0,\, \frac{\pi}{2}$. Then,
		\begin{eqnarray}\label{eq:asymptotic-conditions}
			\lim_{n\to\infty} \frac{\log \| A_n \hat\vartheta_{\infty}\|}{\log \| A_n \|} = -1
			\quad
			\text{if and only if}
			\quad
			\limsup_n \frac{\log |r_n - r_{\infty}|}{\log \| A_n \|} \leq -2.
		\end{eqnarray}
	\end{enumerate}
\end{lemma}
\begin{proof}
	Recall $\vartheta_n$ denotes the unique angle in $(-\pi/2,\pi/2]$ such that $\| A_n \hat\vartheta_n\| = \| A_n \|^{-1}$. Let $\vartheta_n' \in (\vartheta_n-\pi/2, \vartheta_n+\pi/2]$ be the unique angle such that $\| A_n \hat\vartheta_n'\| = \| A_n \|$ and let $\vartheta_n^\perp$ be either $\vartheta_n -\pi/2$ or $\vartheta_n+\pi/2$ and such that $\vartheta_n'$ lies between $\vartheta_n$ and $\vartheta_n^{\perp}$.
	
	Let $\displaystyle v_0 = \begin{pmatrix}
		1 \\0
	\end{pmatrix}$,
	$\displaystyle w_0 = \begin{pmatrix}
		0\\1
	\end{pmatrix}$,
	$v_n=A_n v_0$ and $w_n=A_n w_0$. Then,
	\begin{eqnarray}
		\hat\vartheta_n = \cos\vartheta_n v_0 + \sin\vartheta_nw_0,
	\end{eqnarray}
	and applying $A_n$ on both sides,
	\begin{eqnarray}
		\pm\frac{1}{\|A_n\|} \hat\vartheta_n= A_n \hat\vartheta_n = \cos \vartheta_n v_n + \sin \vartheta_n w_n  \simeq 0,
	\end{eqnarray}
	as $\|A_n\|\to \infty$. Hence,
	\begin{eqnarray}
		r_n = r(A_n) = \frac{\|v_n\|}{\|w_n\|} \simeq |\tan \vartheta_n|.
	\end{eqnarray}
	This shows that, if $\vartheta_n\to\vartheta_\infty$, then $r_n\to|\tan \vartheta_\infty|\in [0,\infty]$. On the other hand, if $r_n\to r_{\infty}\in[0,\infty]$, then $|\vartheta_n| \to \arctan(\rho_{\infty})$, with the convention $\arctan(\infty)=\pi/2$. If $\rho_{\infty}=\infty$, we only get that $|\vartheta_n|\to\pi/2$. Otherwise, it is enough to prove that $\vartheta_n-\vartheta_{n-1}\to0$ to show the convergence of $\vartheta_n$. For this, we use the decomposition
	\begin{eqnarray}\label{eq:decomposition-angles}
		\hat\vartheta = \cos(\vartheta-\vartheta_n) \hat\vartheta_n + \sin(\vartheta-\vartheta_n) \hat\vartheta_n^\perp.
	\end{eqnarray}
	Applied to $\vartheta_{n+1}$, this yields
	\begin{eqnarray}
		\hat\vartheta_{n+1} = \cos(\vartheta_{n+1}-\vartheta_n) \hat\vartheta_n + \sin(\vartheta_{n+1}-\vartheta_n) \hat\vartheta_n^\perp,
	\end{eqnarray}
	and
	\begin{eqnarray}
		A_n\hat\vartheta_{n+1} &=& \pm \| A_n \|^{-1}\cos(\vartheta_{n+1}-\vartheta_n) \hat\vartheta_n + \sin(\vartheta_{n+1}-\vartheta_n) A_n\hat\vartheta_n^\perp\\
			&\simeq&
			\sin(\vartheta_{n+1}-\vartheta_n) A_n\hat\vartheta_n^\perp.\notag
	\end{eqnarray}
	Hence,
	\begin{eqnarray}
		\|A_n\hat\vartheta_n^\perp\| |\sin(\vartheta_{n+1}-\vartheta_n)|
		&\simeq&
		\|A_n\hat\vartheta_{n+1}\|\\
		&\leq&
		\| A_n A_{n+1}^{-1}\| \| A_{n+1} \hat\vartheta_{n+1}\|
		=
		\frac{\| A_{n+1}^{-1}A_n \|}{\| A_{n+1} \|}.
\notag
		\end{eqnarray}
	In the last step, we used that the $A_n$'s are unimodular to switch from $\| A_n A_{n+1}^{-1}\|$ to $\| A_{n+1}^{-1} A_n \|$.
	By condition (ii), it is enough to show that $\|A_n\hat\vartheta_n^\perp\| \gtrsim \|A_n\|$. For this, decompose
	\begin{eqnarray}
		\vartheta_n^\perp = \alpha_n \hat\vartheta_n' + \beta_n \hat\vartheta_n,
	\end{eqnarray}
	for some coefficients such that $\alpha^2 + \beta^2=1$. Note that, by construction, $|\alpha_n|>|\beta_n|$ and hence $|\alpha_n|\geq 1/2$. Therefore,
	\begin{eqnarray*}
		\| A_n\hat\vartheta_n^\perp \|^2 &=& \alpha^2 \| A_n \|^2 + \beta^2 \|A_n\|^{-2} \geq \frac14 \|A_n\|^2.
	\end{eqnarray*}
	This finishes the proof Part 1.
	
	To prove Part 2, assume $\vartheta_n \to \vartheta_{\infty} \in (0,\pi/2)$ and apply \eqref{eq:decomposition-angles} to $\vartheta=\vartheta_{\infty}$ to obtain
	\begin{eqnarray}
		A_n\vartheta_{\infty} = \pm \| A_n \|^{-1} \cos(\vartheta_{\infty}-\vartheta_n) \hat\vartheta_n + \sin(\vartheta_{\infty}-\vartheta_n) A_n \hat\vartheta_n^{\perp}.
	\end{eqnarray}
	Recalling that $\|A_n\|\to\infty$, $\vartheta_n-\vartheta_{\infty}\to0$ and $\|A_n \hat\vartheta_n^\perp\| \asymp \|A_n\|$,
	\begin{eqnarray}
		\frac{\log \|A_n \vartheta_{\infty}\|}{\log \|A_n\|}
		\simeq
		\max \left\{ -1, \frac{\log |\sin(\vartheta_{\infty}-\vartheta_{n})|}{\log\|A_n\|}+1\right\}.
	\end{eqnarray}
	Hence, the left condition in \eqref{eq:asymptotic-conditions} is satisfied if and only if
	\begin{eqnarray}
		\limsup_n \frac{\log |\sin(\vartheta_{\infty}-\vartheta_{n})|}{\log\|A_n\|} \leq -2.
	\end{eqnarray}
	But, disregarding multiplicative constants which will disappear in the limit,
	\begin{eqnarray}
		\sin(\vartheta_{\infty}-\vartheta_n)
		\asymp
		\tan(\vartheta_{\infty}-\vartheta_n)
		\asymp
		\tan\vartheta_{\infty}-\tan\vartheta_n
		=
		r_{\infty} - r_n.
	\end{eqnarray}
\end{proof}
We apply this with $A_n = {\bf T}_{\omega,n}$ to prove Proposition \ref{thm:decay-eigenfunctions}. Define
\begin{eqnarray}
	\begin{pmatrix}
		x^{(1)}_{n+1} \\ x^{(1)}_n
	\end{pmatrix}
	=
	{\bf T}_{\omega,n}
	\begin{pmatrix}
		1 \\ 0
	\end{pmatrix}
	\quad 
	{\text and}
	\quad
	\begin{pmatrix}
		x^{(2)}_{n+1} \\ x^{(2)}_n1
	\end{pmatrix}
	=
	{\bf T}_{\omega,n}
	\begin{pmatrix}
		0 \\ 1
	\end{pmatrix}
\end{eqnarray}
and let $R^{(i)}_n,\, \theta^{(i)}_n,\, n\geq 1$, $i=1,\, 2$ be the corresponding Pr\"ufer radii and phases.
 We let $r_n=R^{(1)}_n/R^{(2)}_n$ and $\vartheta_n$ be as in Lemma \ref{thm:eigendirection}. Then it follows from \eqref{eq:prufer-developped} and some elementary trigonometry that
\begin{eqnarray}
	x^{(1)}_{n+1}x^{(2)}_n-x^{(1)}_nx^{(2)}_{n+1}=R^{(1)}_n R^{(2)}_n  \sin k \sin(\theta^{(2)}_n-\theta^{(1)}_n),
\end{eqnarray}
where $\bar{\theta}^{(i)}_n = nk -\theta^{(i)}_n$.
On the other hand,
\begin{eqnarray}
	x^{(1)}_{n+1}x^{(2)}_n-x^{(1)}_nx^{(2)}_{n+1}
	&=&
	\det\left(
	{\bf T}_{\omega,n}
	\begin{pmatrix}
		1 & 0 \\ 0 & 1
	\end{pmatrix}
	\right)	
	=1.
\end{eqnarray}
This, together with the convergence
\begin{eqnarray}
	\lim_{n\to\infty} \frac{R^{(i)}_n}{\log n} = \beta,\quad i=1,2,
\end{eqnarray}
gives
\begin{eqnarray}\label{eq:limit-sines}
	\lim_{n\to\infty}\frac{\log |\sin(\theta^{(2)}_n-\theta^{(1)}_n)|}{\log n} 
	=\lim_{n\to\infty}\frac{\log |\sin(\bar\theta^{(2)}_n-\bar\theta^{(1)}_n)|}{\log n} 
	= -2\beta.
\end{eqnarray}
Remember the decomposition \eqref{eq:martingale-decomposition}. We have to estimate the difference of the expansions for
$\log R^{(1)}_n$ and $\log R^{(2)}_n$.
By \eqref{eq:limit-sines}, $|\sin(\bar\theta^{(2)}_n-\bar\theta^{(1)}_n)| \lesssim n^{-\beta+\epsilon}$, for any $\epsilon>0$.
Hence, there exists random sequences $(m_n)_n\subset\n$ and $(\Delta_n)_n\subset \R$ such that $\bar\theta^{(1)}_n-\bar\theta^{(2)}_n=m_n\pi + \Delta_n$ and $|\Delta_n|\lesssim n^{-\beta+\epsilon}$. Therefore,
\begin{equation}
	\sin(2\bar\theta^{(2)}_n)
	=
	\sin(2\bar\theta^{(1)}_n+2\Delta_n)
	\simeq
	\sin(2\bar\theta^{(1)}_n)+2\cos(2\bar\theta^{(1)}_n)\Delta_n.
\end{equation}
This shows that 
\begin{eqnarray*}
	\left|V_j \left(\sin(2\bar{\theta}^{(1)}_j) - \sin(2\bar{\theta}^{(2)}_j)\right)\right| 
	&\lesssim& j^{-\frac{1}{2}-2\beta + \epsilon},\\
	\text{and}
	\quad
	\left|V_j^2 \left(\cos^2(\bar{\theta}^{(1)}_j) - \cos^2(\bar{\theta}^{(2)}_j)\right)\right| 
	&\lesssim&
	 j^{-1-2\beta + \epsilon},
\end{eqnarray*}
by a similar argument. Hence, 
\begin{eqnarray}\label{eq:decomposition-radius}
	\log r_n 
	&=& 
	-\sum^n_{j=1} \frac{V_{\omega,j}}{\sin k} \left(\sin(2\bar{\theta}^{(1)}_j) - \sin(2\bar{\theta}^{(2)}_j) \right)
	+
	\sum^n_{j=1} A_j
\end{eqnarray}
 where the first sum is a convergent martingale by Lemma \ref{thm:martingales} with $\gamma=\frac12 + 2\beta -\epsilon$ and the second one is absolutely convergent as $A_j=O(j^{-1-2\beta+\epsilon})$. This shows that $r_n\to r_{\infty}\in(0,\infty)$ almost surely which implies that $\vartheta_n$ has a limit $\vartheta_{\infty}\neq 0,\, \pi/2$ by the first part of Lemma \ref{thm:eigendirection}. 

The equivalence \eqref{eq:asymptotic-conditions} in our context corresponds to
\begin{eqnarray}
	\lim_{n\to\infty}\frac{\log R_n( \vartheta_{\infty})}{\log n}=-\beta
	\quad 
	\text{if and only if}
	\quad
	\limsup_n \frac{\log |r_n-r_{\infty}|}{\log n} \leq -2\beta.
\end{eqnarray}
Let us denote by $\log r_n = M_n + S_n$ the decomposition \eqref{eq:decomposition-radius} and $\log r_{\infty}=M_{\infty}+S_{\infty}$ where $M_{\infty}$ and $S_{\infty}$ are the almost sure limits of $M_n$ and $S_n$ respectively. Then,
\begin{eqnarray}
	|r_{\infty}-r_n|
	&=&
	e^{M_{\infty}+S_{\infty}}\left| 1-\e^{M_n-M_{\infty}+S_n-S_{\infty}}\right|
	\simeq
	\e^{M_{\infty}+S_{\infty}}\left| M_n-M_{\infty}+S_n-S_{\infty}\right|\\
	&\lesssim&
	\e^{M_{\infty}+S_{\infty}} n^{-2\beta+2\epsilon},
\end{eqnarray}
by the last statement of Lemma \ref{thm:martingales} with $\gamma = \frac12 + 2\beta - \epsilon$.
This finishes the proof of Proposition \ref{thm:decay-eigenfunctions}.


\subsection{Analysis of the Pr\"ufer phases}\label{app:phases}

We begin with a simple observation: from \eqref{eq:main-rec-complex},
\begin{eqnarray}
	\left| e^{i(\theta_{n+1}-\theta_n)}-1\right| = \left| \frac{\rho_{n+1}}{\rho_n}-1\right|\leq \frac{| V_{\omega,a}|}{|\sin k|}\leq 1,
\end{eqnarray}
for $n$ large enough. Hence, for $n$ large enough, $|\theta_{n+1}-\theta_n|< \pi/2$ and
\begin{eqnarray}
	|\theta_{n+1}-\theta_n| \leq \frac{\pi}{2} |\sin(\theta_{n+1}-\theta_n)| 
	\leq \frac{\pi}{2} \left|e^{i(\theta_{n+1}-\theta_n)}-1\right| 
	\lesssim n^{-\alpha}.
\end{eqnarray}
This can be written in the equivalent form 
\begin{eqnarray}\label{eq:small-phases}
	|\bar{\theta}_{n+1}-\bar{\theta}_n-k|\leq c_0 n^{-\alpha},
\end{eqnarray}
for some $c_0>0$,
which will be more suitable for our purposes.

The next lemma provides the control of the Pr\"ufer phases needed to complete the proof of Proposition \ref{thm:lyapuvov}.
\begin{lemma}
Let $0<\alpha\leq \frac12$.
Let $E\in (-2,2)$ corresponding to values of $k$ different from $\frac{\pi}{4}$, $\frac{\pi}{2}$ and $\frac{3\pi}{4}$. Then,
\begin{eqnarray}
	\limsup_{n\to\infty}
	\frac{\sum^n_{j=1}\esp[ V_{\omega,j}^2]\left( \frac{1}{2}\cos 2 \bar \theta_j + \frac{1}{4} \cos 4 \bar \theta_j\right)}{\sum_{j=1}^n j^{-2\alpha}}
	=
	0.
\end{eqnarray}
\end{lemma}
\begin{proof}
	The key to the proof is \cite[Lemma 8.5]{KLS} which states: suppose that $y\in\R$ is not in $\pi \Z$. Then, there exists a sequence of integers $q_l \to \infty$ such that
	\begin{eqnarray}\label{eq:trigonometrics}
		\left|\sum^{q_l}_{j=1} \cos \theta_j \right| \leq 1 + \sum^{q_l}_{j=1} \left| \theta_j-\theta_0 - jy \right|,
	\end{eqnarray}
	for all $(\theta_j)_{j\geq 0}\subset \R$.
	
	We will treat the term with $\cos(4\bar{\theta}_j)$ as the other one is similar. We will take $y=4k$ above. 
	Let $n$ be large enough so that it can be written as 	$n=n_0+Kq_l$ with $n_0\ge q_l^2$ and $4c_0 n_0^{-\alpha}\le q_l^{-2} $ (where $c_0$ is the constant from \eqref{eq:small-phases}). Then,
\begin{eqnarray}
	 \left|\sum_{j=n_0+1}^n j^{-2\alpha}\cos(4\bar\theta_j)\right|
	 &=&
	 \left|\sum_{m=0}^K \sum_{r=1}^{q_l}(n_0+mq_l+r)^{-2\alpha}\cos(4\bar\theta(n_0+mq_l+r))\right|
	 \\
	&\leq&	 
	\sum_{m=0}^K(n_0+mq_l)^{-2\alpha} \left|\sum_{r=1}^{q_l}\cos(4\bar\theta(n_0+mq_l+r))\right|
	\\
	&+&
	\sum_{m=0}^K \sum_{r=1}^{q_l} \left|(n_0+mq_l+r)^{-2\alpha}-(n_0+mq_l)^{-2\alpha}\right|
	\\
	&=:&
	A + B.
\end{eqnarray}
By \eqref{eq:trigonometrics},
\begin{eqnarray}
	A
	&\leq&
	\sum_{m=0}^{K}(n_0+mq_l)^{-2\alpha}\left(1+4\sum_{r=1}^{q_l}|\bar\theta(n_0+mq_l+r)-\bar\theta(n_0+mq_l)-kr|\right).
\end{eqnarray}
Now, by \eqref{eq:small-phases},
\begin{eqnarray}
	&&
	4\sum_{r=1}^{q_l}|\bar\theta(n_0+mq_l+r)-\bar\theta(n_0+mq_l)-kr|
	\leq
	c_0 \sum_{r=1}^{q_l}\sum_{s=1}^r(n_0+mq_l+r)^{-\alpha}
	\\
	&&
	\phantom{blablablabla}
	\leq
	c_0(n_0+mq_l)^{-\alpha}\sum_{r=1}^{q_l} r
	\leq
	c_0 q_l^2 n_0^{-\alpha}
	\leq
	1.
\end{eqnarray}
Thus,
\begin{eqnarray*}
	A
	\leq
	2
	\sum_{m=0}^{K}(n_0+mq_l)^{-2\alpha}
	\leq
	2 q_l^{-2\alpha} \sum_{m=0}^{K}(n_0q_l^{-1}+m)^{-2\alpha}
	\leq 
	c_1 q_l^{-2\alpha} K^{1-2\alpha}
	\leq
	c_1 q_l^{-1} n^{1-2\alpha},
\end{eqnarray*}
for some finite $c_1>0$. To estimate $B$, we use that
\begin{eqnarray}
	\left| (n_0+mq_l+r)^{-2\alpha}-(n_0+mq_l)^{-2\alpha}\right|
	\leq
	c_2 (n_0+mq_l)^{-2\alpha-1} r,
\end{eqnarray}
for some finite $c_2>0$ which shows that
\begin{eqnarray}
	B
	&\leq&
	c_2 \sum_{m=0}^K \sum_{r=1}^{q_l} (n_0+mq_l)^{-2\alpha-1} r
	\leq 
	c_2 q_l^2 n_0^{-1}
	\sum_{m=0}^K
	(1+n_0^{-1}mq_l)^{-1}(n_0+mq_l)^{-2\alpha}
	\\
	&\leq&
	c_2
	\sum_{m=0}^K
	(n_0+mq_l)^{-2\alpha},
\end{eqnarray}
where we used $q_l^2 n_0^{-1}\leq 1$. This last sum can be estimated as above. 
Summarizing,
\begin{eqnarray}
	\left|\sum^{q_l}_{j=1} \cos \theta_j \right|
	\leq
	c_3 q_l^{-1} n^{1-2\alpha},
\end{eqnarray}
for some finite $c_3>0$. This finishes the proof.
\end{proof}



\begin{thebibliography}{99}

\bibitem{AM} M. Aizenman, S. Molchanov, \textit{Localization at large disorder and an externe energies:
an elementary derivation}, Comm. Math. Phy. {\bf 157}, 245-278 (1993).

\bibitem{ASW} M. Aizenman, R. Sims, S. Warzel, \textit{Stability of the asbolutely continuous spectrum of random Schr\"odinger operators on tree graphs},
 Probab. Theory Related Fields {\bf 136}, 363-394 (2006).
 
\bibitem{AW} M. Aizenman, S. Warzel, \textit{Random Operators: Disordered effects on Quantum spectra and dynamics}, Graduate Studies in Mathematics, vol 168 AMS (2016).

\bibitem{Azuma} Azuma, K., \textit{Weighted sums of certain dependent random variables}, T\^ohoku Mathematical Journal {\bf 19}, 3, 357–367 (1967).


\bibitem{B1} J. Bourgain, \textit{ On random Schr\"odinger operators on $\Z^2$},
 Discret Contin. Dyn. Syst. {\bf 8}, 1-15 (2002).

\bibitem{B2} J. Bourgain, \textit{Random lattice schr\"odinger operators with decaying potential: some higher dimensional phenomena}, 
Geometric Aspects of Functional Analysis, Lectures Notes in Math. {\bf 1807}, 70-98, Springer, Berlin-Heidelberg (2003).


\bibitem{BDFG} V. Bucaj, D. Damanik, J. Fillman, V. Gerbuz, T. VandenBoom, F. Wang, Z. Zhang,
\textit{Localization for the one-dimensional Anderson model via positivity and large deviations for the Lyapunov exponent},
Trans. Amer. Math. Soc. {\bf 372}, 3619-3667 (2019),

\bibitem{BMT02} O. Bourget, G. Moreno, A. Taarabt,
\textit{Dynamical localization for the one-dimensional continuum Anderson model in a decaying random potential}, arXiv:2001.02197.

\bibitem{BMT01} O. Bourget, G. Moreno, A. Taarabt,
\textit{One-dimensional Discrete Dirac Operators in a Decaying Random Potential I: Spectrum and Dynamicsl}, arXiv:200102199.

\bibitem{Car} R. Carmona, \textit{Exponential localization in one dimensional disordered systems}, Duke Math. J. {\bf 49} 191-213 (1982).


\bibitem{CFKS} H.L. Cycon, R.G. Richard, W. Kirsch, B. Simon, \textit{Schr\"odinger operators with applications to quantum Mechanics and global Geometry}, Texts and Monographs ib Physics, Springer Study Edition, Springer-Verlag, Berlin, (1987).

\bibitem{CKM} R. Carmona, A. Klein, F. Martinelli, \textit{Anderson Localization for Bernoulli and Other Singular Potentials}, Commun. Math. Phys. {\bf 108}, 41-66 (1987).

\bibitem{DG} D. Damanik, A. Gorodetski,
\textit{An extension of the Kunz-Souillard approach to localization in one dimension and applications to almost-periodic Schr\"odinger operators},
Adv. Math. (2016).

\bibitem{DS} D. Damanik, G. Stolz, \textit{A continuum version of the Kunz–Souillard approach to localization in one dimension}, Journal f\"ur die reine und angewandte Mathematik (Crelles Journal), {\bf 660}, 99-130 (2011).

\bibitem{D} F. Delyon, \textit{Appeareance of a purely singular continuous spectrum in a class of random Schr\"odinger operators}
J. Statist. Phys. {\bf 40}, 621-630 (1985).

\bibitem{DSS} F. Delyon, B. Simon, B. Souillard,
\textit{From power pure point to continuous spectrum in disordered systems},
Ann. Henri Poincaré, {\bf 42} vol. 6, 283-309 (1985).

\bibitem{DeRJLS1} R. Del Rio, S. Jitomirskaya, Y. Last and B. Simon, \textit{What is localization?}, Phys. Rev. Lett. 75 117-119 (1995).

\bibitem{DeRJLS2} R. Del Rio, S. Jitomirskaya, Y. Last and B. Simon, \textit{Operators with singular contin-
uous spectrum IV: Hausdorff dimensions, rank one pertubations and localization}, J. Anal. Math. 69 153-200 (1996).


\bibitem{Durrett} Durrett, R., \textit{Probability: theory and examples}, Cambridge Series in Statistical and Probabilistic Mathematics, Fourth Edition, Cambridge University Press,  New York (2010).


\bibitem{FGKM} A. Figotin, F. Germinet, A. Klein, P. M\"uller, \textit{Persistence of Anderson localization in Schr\"odinger operators with decaying random potentials}, Ark. Mat. {\bf 45} 15-30 (2007).


\bibitem{FHS} R. Froese, D. Hasler, W. Spitzer, \textit{Absolutely continuous spectrum for the Anderson model on a tree: a geometric proof of Klein's theorem},
Comm. Math. Phys. {\bf 269}, 239-257 (2007).

\bibitem{GKT} F. Germinet, A. Kiselev, S. Tcheremchantsev,
\textit{Transfer matrices and transport for Schr\"odinger operators},
Ann. Inst. Fourier {\bf 54}, 787-830 (2004).

\bibitem{GK1} F. Germinet, A. Klein, \textit{Bootstrap multiscale analysis and localization in random
media}, Commun. Math. Phys. {\bf 222}, 415-448 (2001).

\bibitem{GKS} F. Germinet, A. Klein, J. Schenker, \textit{Dynamical delocalization in
random Landau Hamiltonians}, Ann. Math. {\bf 166} 215-244 (2007).

\bibitem{GT} F. Germinet, A. Taarabt, 
\textit{Spectral  properties  of  dynamical  localization  for Schr\"odinger operators}, Rev. Math. Phys. 25 9 (2013).


\bibitem{GMP} I. Goldsheid, S. Molchanov, L. Pastur, \textit{A pure point spectrum of the stochastic
one-dimensional Schr\"odinger equation}, Funct. Anal. Appl. {\bf 11}, 1-10 (1977).


\bibitem{GZ} L. Ge, X. Zhao, 
\textit{Exponential dynamical localization in expectation foe one dimensional Anderson model},
J. Spect. Theory

\bibitem{JL} V. Jak\u{s}i\'{c}, Y, Last,
\textit{Spectral structure of Anderson type Hamiltonians},
Inven. Math. {\bf 141}, 561-577 (2000).


\bibitem{JZ} S. Jitomirskaya, X. Zhu,
\textit{Large deviations of the Lyapunov exponent and localization for the 1D Anderson model}, 
Comm. Phys. Math. (2019).

\bibitem{Kl} A. Klein, \textit{Extended states in the anderson model on the bethe lattice}, Advances in Mathematics {\bf 133} 163-184 (1998).

\bibitem{Kr} M. Krishna, \textit{Anderson model with decaying randomness: existence of extended states}, Proc. Indian Acad. Sci. (Math. Sci.) {\bf 100}
285-294 (1990).

\bibitem{KKO} W. Kirsch, M. Krishna, J. Obermeit,
\textit{Anderson model with decaying randomness: mobility edge},
Math. Z. {\bf 235}, 421-433 (2000).


\bibitem{KLS} A. Kiselev, Y. Last, B. Simon,
\textit{Modified Pr\"ufer and EFGP transforms and the spectral analysis of one-dimensional schr\"odinger operators},
Comm. Math. Phys. {\bf 194}, 1-45 (1998).

\bibitem{KRS} A. Kiselev, C. Remling, B. Simon, \textit{Effective perturbation methods for one-dimensional Schr\"odinger operators}, J. Differential Equations {\bf 151}, 290-312 (1999).

\bibitem{KS} H. Kunz, B. Souillard, \textit{Sur le spectre des op\'erateurs aux diff\'erences finies al\'eatoires}, Comm. Math. Phys. {\bf 78}, 201-246 (1980).


\bibitem{KU} S. Kotani, N. Ushiroya, 
\textit{One-dimensional Schr\"odinger operators with random decaying potentials},
Comm. Math. Phys. {\bf 115}, 247-266 (1988).

\bibitem{LS} Y. Last, B. Simon, 
\textit{Eigenfunctions, transfer matrices, and absolutely continuous spectrum of one-dimensional Schrödinger operators}, Invent. Math. {\bf 135}, 329 (1999).

\bibitem{Si82} B. Simon, 
\textit{Some Jacobi matrices with decaying potential and dense point spectrum}, Comm. Math. Phys. (1982).

\bibitem{Si95} B. Simon, \textit{Spectral Analysis of rank one perturbations and applications},
CRM Lectures Notes, Vol. 8, Amer. Math. Soc, Providence, RI, 109-149 (1995).

\end{thebibliography}

\end{document}